\PassOptionsToPackage{pdfpagelabels=false}{hyperref} 
\documentclass[a4paper,11pt]{scrartcl}
\usepackage[colorlinks=TRUE,linkcolor=blue,citecolor=green]{hyperref} 

\usepackage{a4wide}
\usepackage[german, english]{babel}
\usepackage[utf8]{inputenc}
\usepackage{graphicx} 

\usepackage{amsfonts}
\usepackage{amsmath}
\usepackage{amssymb}
\usepackage{amsthm}

\usepackage{tocloft}
\usepackage[numbers]{natbib} 
\usepackage{setspace}	

\usepackage[format=plain,hangindent=.3 cm]{caption}

\newtheorem{theorem}{Theorem}

\newtheorem{lemma}[theorem]{Lemma}

\theoremstyle{definition}
\newtheorem*{remark}{Remark}

\makeatletter
\g@addto@macro{\thm@space@setup}{\thm@headpunct{:}}
\makeatother

\newcommand{\e}{\mathrm{e}} 
\newcommand{\dd}[1]{\, \mathrm{d} #1}
\newcommand{\ex}[1]{\, \mathrm{exp} \left( #1 \right)}
\newcommand{\1}[1]{\, \mathbb{I} \left\{ #1 \right\}}
\renewcommand{\i}{{\mathbf{\mathfrak{i}}}}
\newcommand{\EV}[2][\!]{\mathbb{E}^{#1} \left[ #2 \right]} 

\newcommand{\logn}[1]{\, \ln \left( #1 \right)}
\newcommand{\sign}[1]{\, \mathrm{sgn} \left( #1 \right)}

\newcommand{\T}{\cdot}

\DeclareMathOperator{\PM}{\mathbb{P}}
\DeclareMathOperator{\QM}{\mathbb{Q}}

\DeclareMathOperator{\R}{\mathbb{R}}
\DeclareMathOperator{\C}{\mathbb{C}}
\DeclareMathOperator{\F}{\mathcal{F}}

\DeclareMathOperator{\Flt}{{\mathrm{Flt}}}

\newcommand{\p}{\  .}

\newcommand{\titleinfo}{Affine LIBOR models driven by real-valued affine processes} 
\newcommand{\authorinfo}{Wolfgang Müller, Stefan Waldenberger} 
\newcommand{\keywords}{LIBOR rate models, forward price models, affine processes, volatility smile}

\begin{document}
\title{\titleinfo}


\AtEndDocument{\bigskip \bigskip{\footnotesize%
  \textsc{Graz University of Technology, Institute of Statistics, NAWI Graz} \par
  Kopernikusgasse 24/III, 8010 Graz, Austria \par
  \textit{E-mail address: }{w.mueller@tugraz.at}, {stefan.waldenberger@tugraz.at} \par
}}

\selectlanguage{english}




\begin{spacing}{1}

\thispagestyle{plain}
\begin{center}
\begin{minipage}{.8 \textwidth}
\begin{center}
{\Large \bf \titleinfo} \\ \bigskip
{\large \textsc{\authorinfo}} \\  \smallskip
\end{center}
  \bigskip \\
\textsc{Abstract:} 
The class of affine LIBOR models is appealing since it satisfies three central requirements of interest rate modeling. It is arbitrage-free, interest rates are nonnegative and caplet and swaption prices can be calculated analytically. In order to guarantee nonnegative interest rates affine LIBOR models are driven by nonnegative affine processes, a restriction, which makes it hard to produce volatility smiles. We modify the affine LIBOR models in such a way that real-valued affine processes can be used without destroying the nonnegativity of interest rates. Numerical examples show that in this class of models pronounced volatility smiles are possible. 
\end{minipage}
\vspace{.5 cm}
\end{center}

\section{Introduction}
Market models, the most famous example being the LIBOR market model, are very popular in the area of interest rate modeling. If these models generate nonnegative interest rates they usually do not give semi-analytic formulas for both basic interest rate derivatives, caps and swaptions. One exception is the class of affine LIBOR models proposed by \citet{KPT11}. Using nonnegative affine processes as driving processes affine LIBOR models guarantee nonnegative forward interest rates and lead to semi-analytical formulas for caps and swaptions, so that calibration to interest rate market data is possible.

This paper modifies the setup of \citet{KPT11} to allow for not necessarily nonnegative affine processes. This modification still leads to semi-analytical formulas for caps and swaptions and guarantees nonnegative forward interest rates, but allows for a wider class of driving affine processes and hence is more flexible in producing interest rate skews and smiles.
\citet{DGG12} also propose a modification of affine LIBOR models. There driving processes are affine processes with values in the space of positive semidefinite matrices. The approach in this paper has the advantage that a flexible class of implied volatility surfaces can be produced with a much smaller number of parameters.

The structure of this paper is as follows. In section \ref{sec:affineprocess} affine processes and their properties are reviewed. Section \ref{sec:marketmodels} introduces the necessary notation and market setup and reviews affine LIBOR models. It concludes with some comments on practical implementation. Section \ref{sec:modaffinemodel} is the main section of this paper. The first part presents the modified affine LIBOR model and semi-analytical pricing formulas for caps and swaptions are derived. The second part 
then gives some examples of usable affine processes with numerical calculations.

\section{Affine processes} \label{sec:affineprocess}
Let  $X = (X_t)_{0 \leq t \leq T}$ be a homogeneous Markov process with values in $D = \R^m_{\geq 0} \times \R^n$ realized on a measurable space $(\Omega,\mathcal{A})$ with filtration  $(\F_t)_{0 \leq t \leq T}$, with regards to which $X$ is adapted. Denote by $\PM^x[\cdot]$ and $\EV[x]{\cdot}$ the corresponding probability and expectation when $X_0 = x$. 
X is said to be an affine process, if its characteristic function has the form
\begin{equation}  \label{eq:affineMomentGen}
\EV[x]{ \e^{u\T X_{t}}}  = \ex{\phi_{t}(u) + \psi_{t}(u)\T x}, \quad u \in \i \R^d, x \in D,
\end{equation}
where $\phi: [0,T]\times \i \R^d  \rightarrow \C$ and $\psi: [0,T] \times \i \R^d  \rightarrow \C^d$ with $\i \R^d = \{u \in \C^d: \mathrm{Re}(u) = 0 \}$ and $\cdot$ denoting the scalar product in $\R^d$. 
By homogeneity and the Markov property the conditional characteristic function satisfies
$$ \EV[x]{e^{u\T X_{t}} \vert \F_s} = \ex{\phi_{t-s}(u) + \psi_{t-s}(u)\T X_s}. $$
Accordingly affine processes can also be defined for inhomogeneous Markov processes (see \citet{FI05}), in which case the above equality reads
\begin{equation*} \label{eq:affineMomentGeninhom}
\EV[x]{ \e^{u\T X_t} \vert \F_s} = \ex{\phi_{s,t}(u) + \psi_{s,t}(u)\T X_s}, \quad u \in \i \R^d, x \in D, 
\end{equation*}
with $\phi_{s,t}: \i \R^d  \rightarrow \C$ and $\psi_{s,t}: \i \R^d  \rightarrow \C^d$ for $0 \leq s \leq t$.

$X$ is called an analytic affine process (see \citet{KR08}), if $X$ is stochastically continuous and the interior of the set\footnote{$\mathcal{V}$ can be described as the (convex) set, where the extended moment generating function of $X_t$ is defined for all times $t \leq T$ and all starting values $x \in E$. By Lemma 4.2 in \citet{KM11} the set $\mathcal{V}$ is in fact equal to the seemingly smaller set 
$\left \{ u \in \C^d:   \exists x \in \mathrm{int}(D): \EV[x]{ \e^{\mathrm{Re}(u)\T X_T}} < \infty \right \}.$}
\begin{align}
\mathcal{V} & := \left \{ u \in \C^d: \sup_{0 \leq s \leq T}  \EV[x]{ \e^{\mathrm{Re}(u)\T X_s}} < \infty \quad  \forall x \in D \right \}, \label{eq:momset}
\end{align}
contains 0\footnote{This also implies that $X$ is conservative, i.e. $\PM^x(X_t \in D) = 1 \  \forall x \in D \text{ and } 0 \leq t \leq T$. 
}. In this case the functions $\phi$ and $\psi$ have continuous extensions to $\mathcal{V}$, which are analytic in the interior, such that \eqref{eq:affineMomentGen} holds for all $u \in \mathcal{V}$.

The class of affine processes includes Brownian motion and more generally all Lévy processes. Since Lévy processes have stationary independent increments, in this case $\psi_t(u) = u$ and $\phi_t(u) = t \kappa(u)$, where $\kappa$ is the cumulant generating function of the Lévy process. Ornstein-Uhlenbeck processes are further important examples of affine processes. They are discussed in section \ref{sec:OUprocess}. 

The standard reference for affine processes is \citet{DFS03}. There they give a characterization of affine processes, where $\phi$ and $\psi$ are specified as solutions of a system of differential equations\footnote{The fact that this characterization holds for all stochastically continuous affine processes was first shown in \citet{KST11} and later for affine processes with more general state spaces in \citet{KST11b} and \citet{CT13}.}. 
Of all the rich theory of affine processes the methods in this paper only use the specific form \eqref{eq:affineMomentGen} of their moment generating function and the following property.
\begin{lemma} \label{lem:psimontonicity}
Let $X$ be a one-dimensional analytic affine process and $\mathrm{Re}(u) < \mathrm{Re}(w), u,w \in \mathcal{V}$. Then
$\mathrm{Re}(\psi_t(u)) < \psi_t(\mathrm{Re}(w)),$
i.e. $\psi_t \vert_{\mathcal{V} \cap \R}$ is strictly increasing. 
\end{lemma}
\begin{proof}
The case $D = \R_+$ is already contained in \citet{KPT11}. In case $D = \R$ the lemma follows from the fact that by Proposition 3.3 in \citet{KST11} $\psi_t(u) = \e^{\beta t} u$ for some constant $\beta$.
\end{proof}
%
%
\begin{remark}
If $D = \R_+$, it is known that both, $\psi$ and $\phi$, are monotonically increasing (\citet{KPT11}). 
With $D=\R$ this stays true for $\psi$, but not $\phi$, as the deterministic affine process $X_t = x_0 - t$ shows. 
\end{remark}

\section{Interest rate market models} \label{sec:marketmodels}

\subsection*{Classical market models}
Consider a tenor structure $0 < T_1 < \dots < T_N < T_{N+1} =: T$ and a market consisting of zero coupon bonds with maturities $T_1,\dots,T_{N+1}$. Their price processes $(P(t,T_k))_{0 \leq t \leq T_k}$ are assumed to be nonnegative semimartingales on a filtered probability space $(\Omega, \mathcal{A}, (\mathcal{F}_t)_{0 \leq t \leq T}, \PM)$, which satisfy $P(T_k,T_k) = 1$ almost surely.  If there exists an equivalent probability measure $\QM^{T}$ such that the normalized bond price processes $P(\cdot,T_k) / P(\cdot,T)$ are martingales\footnote{One can extend bond price processes to $[0,T]$ by setting $P(t,T_k) := \frac{P(t,T)}{P(T_k,T)}$ for $t > T_k$, so that $P(\cdot,T_k) / P(\cdot,T)$ is a martingale on $[0,T]$ if and only if it is a martingale on $[0,T_k]$. Economically this can be interpreted as immediately investing the payoff of a zero coupon bond into the longest-running zero coupon bond.
}, the market is arbitrage-free. In this case we can define equivalent martingale measures $\QM^{T_k}$ for the numeraires $P(t,T_k)$ instead of $P(t,T)$ by
\begin{equation} \frac{\dd{\QM^{T_k}}}{\dd{\QM^{T}}} =  \frac{1}{P(T_k,T)} \frac{P(0,T)}{P(0,T_k)}.  \label{eq:measurechange}
\end{equation}
In particular under the measure $\QM^{T_k}$ the forward bond price process ${P(\cdot,T_{k-1})}/{P(\cdot,T_{k})}$ and the forward interest rate process $F^k(\cdot)$,
\begin{equation} \label{eq:fwdintrate}
F^k(t) = \frac{1}{\Delta_k} \left(\frac{P(t,T_{k-1})}{P(t,T_{k})} - 1\right), \qquad \Delta_k = T_{k}-T_{k-1},
\end{equation}
are martingales. This is the basic market setup used throughout the rest of the paper.

In the classical LIBOR market models forward interest rate processes $F^k$ are modeled as continuous exponential martingales under their respective martingale measure $\QM^{T_k}$. Hence forward interest rates are positive. Using driftless geometric Brownian motions as driving processes caplet prices are given by the Black formula (\citet{BF76}) while swaption prices cannot be calculated analytically.
Alternatively one can start with modeling the forward bond price processes ${P(\cdot,T_{k-1})}/{P(\cdot,T_{k})}$ instead of forward interest rate processes. Using again exponential martingales like a driftless Brownian motion it is then 
possible to analytically calculate caplet and swaption prices (see \citet{EO05}). The drawback of this approach is that forward interest rates will be negative with positive probability.

\citet{KPT11} proposed the affine LIBOR models, where forward interest rates are nonnegative while swaption and caplet prices can still be calculated semi-analytically, i.e. up to a numerical integration. The above approaches model the individual forward interest rate processes (resp. forward bond price process) with respect to the individual measure $\QM^{T_k}$ under which they are a martingale. Contrary \citet{KPT11} model the price processes $P(\cdot,T_k)/P(\cdot,T)$, which are all martingales under the same probability measure $\QM^{T}$.
\begin{remark}
Note that all models mentioned in this paper do not fully specify the whole term structure, but only part of it. In order to price derivatives not contained within the specified tenor structure it is necessary to specify some kind of interpolation scheme. Arbitrary interpolations may lead to arbitrage, however one can always choose an interpolation method, such that the model stays arbitrage-free (\citet{WE10}).
\end{remark}

\subsection*{The affine LIBOR models}

This section presents a summary of the affine LIBOR model introduced in \citet{KPT11}. 
On the filtered probability space $(\Omega, \mathcal{A}, (\mathcal{F}_t)_{0 \leq t \leq T}, \QM^T)$ consider a nonnegative analytic affine process $X$ with a fixed starting value $x_0 \in \R_{\geq0}^d$. For the tenor structure  $0 < T_1 < \dots < T_N < T_{N+1} =: T$ define for $k=1,\dots,N$ and $0 \leq t \leq T_k$
\begin{equation} \label{eq:affinefwdprices}
 \frac{ P(t,T_k)}{P(t,T)} := \EV[\QM^T]{\e^{u_k\T X_T} \vert \F_t} = \e^{\phi_{T-t}(u_k) + \psi_{T-t}(u_k)\T X_t}, \qquad u_k \geq 0, u_k \in \mathcal{V},
\end{equation}
where $\EV[\QM]{\cdot}$ denotes the expectation with respect to a probability measure\footnote{Since $x_0$ is fixed, contrary to to section \ref{sec:affineprocess} any dependence of probability measures on the starting value of the Markov process $X$ will be suppressed from now on.}  $\QM$. These price processes are martingales and the resulting model is arbitrage-free.

Writing
\begin{equation} \label{eq:fwdnormbonds}
\frac{P(t,T_{k-1})}{P(t,T_k)} = \left. {\frac{P(t,T_{k-1})}{P(t,T)}} \right / { \frac{P(t,T_k)}{P(t,T)}} 
\end{equation}
in \eqref{eq:fwdintrate} shows that forward interest rates being nonnegative is equivalent to normalized bond prices of \eqref{eq:affinefwdprices} satisfying
\begin{equation} \label{eq:bondmonotonicity}
\frac{ P(t,T_1)}{ P(t,T)} \geq .. \geq \frac{ P(t,T_N) }{P(t,T)} \geq 1.
\end{equation}
Since for $x\geq 0$, $\e^{u^T x}$ is monotonically increasing in every component of $u$ the monotonicity for normalized bond prices in \eqref{eq:bondmonotonicity} is satisfied as long as $u_1 \geq \dots u_{N} \geq 0$. 

The parameters $u_k$ in \eqref{eq:affinefwdprices} should be determined, so that the starting values of normalized bond prices ${ P(0,T_k)}/{P(0,T)} =
 \ex{\phi_T(u_k) + \psi_T(u_k)\T x_0}$
fit the initial term structure inferred from actual market data. 
For most affine processes every term structure can be fitted and for currently nonnegative forward interest rates this can be done using an decreasing sequence $u_1 \geq \dots \geq u_N \geq 0$ (see \citet{KPT11}).

\begin{remark}
Since $X$ is nonnegative, the k-th normalized bond price is not only greater equal to one, but is bounded from below by the time-dependent constant
$ \ex{\phi_{T-t}(u_{k}) },$
which is strictly greater than one. Accordingly in the affine LIBOR models forward interest rates are bounded from below by a strictly positive time-dependent constant.
\end{remark}

Affine LIBOR models lead 
to nonnegative forward interest rates. Additionally this specification is appealing because the density processes for changes of measures are again exponentially affine in $X_t$, i.e. inserting \eqref{eq:affinefwdprices} into \eqref{eq:measurechange} gives
\begin{align*}
\frac{\dd{\QM^{T_k}}}{\dd{\QM^{T}}} & = \frac{P(0,T)}{P(0,T_k)}  \e^{\phi_{T-T_k}(u_k) + \psi_{T-T_k}(u_k)\T X_{T_k}}.
\end{align*}
Moreover normalized bond prices and because of \eqref{eq:fwdnormbonds} also forward bond prices are of exponential affine form. It follows that the moment generating function of the logartihm of normalized bond prices under $\QM^{T_k}$ is also of exponential affine form and that calculation of caplet prices is possible via a one-dimensional Fourier inversion. 
If the dimension of the driving process is one, swaption prices can also be calculated via one-dimensional Fourier inversion (see \citet{KPT11}). Hence this approach satisfies both, nonnegative interest rates and analytical tractability of standard interest rate market instruments. If the dimension is larger than one, the exact price of swaptions can only be calculated via higher-dimensional integration, the dimension of which is the length of the underlying swap. Alternatively \citet{GPSS14} provide approximate formulas for swaptions. 

\subsubsection*{Practical application of the affine LIBOR model} \label{sec:practicalapplication}
Although this framework is elegant from a theoretical point of view, a practical implementation faces several difficulties which shall be discussed here. 

First, calibration of interest rates and implied volatilities cannot be separated. The initial term structure can be fitted using the $u_k$, but the parameters $u_k$ also have a strong impact on implied volatilities. This can be seen by looking at the forward bond price
\begin{equation} \frac{1}{P(T_{k-1},T_k)} =  \ex{\phi_{T-t}(u_{k-1})-\phi_{T-t}(u_{k})+(\psi_{T-t}(u_{k-1})-\psi_{T-t}(u_k))\T X_{T_{k-1}}}, \label{eq:KRTfltrv} \end{equation}
which is the random variable responsible for the payoff of a caplet.
The driving process $X$ influences the distribution of this random variable through two different channels. First via the parameters of the driving process itself and second via the parameters $u_k$ (depending on $X$ and the initial interest rate term structure). Hence for changes in the yield curve different parameters are required to reproduce the same implied volatility surface. If $X$ is a Lévy process, then as mentioned in section \ref{sec:affineprocess} $\psi_t(u) = u$ and it follows that the distribution of \eqref{eq:KRTfltrv} depends on the difference $u_{k+1} - u_k$, which in turn is related to the steepness of the initial yield curve\footnote{This is similar for most affine processes, but is best visible for Lévy processes.}. Hence caplet implied volatilities are especially sensitive with regards the steepness of the initial yield curve.

Second, interest rates and volatilities of this model depend on the final horizon $T$. Changing the horizon $T$ while using the same affine process $X$ will lead to different results and
there is no general way of rescaling the parameters of $X$ to negate such an effect. This is rather counterintuitive, since extending the horizon of a model should not change the results for quantities already included with the shorter horizon. 
 
Third, the types of possible volatility surfaces is rather constrained in the fully analytically tractable one-dimensional case. For example, we were only able to generate volatility skews\footnote{The smile example of \citet{KPT11}, figure 9.2, using an Ornstein-Uhlenbeck process seems to be numerically incorrect for strikes smaller than 0.4. With the mentioned initial yield curve the underlying interest rate is always larger than the strike, which corresponds to a zero implied volatility, destroying the displayed smile.}.
This might be resolved by using higher-dimensional nonnegative processes. However, in multidimensional affine LIBOR models swaptions can no longer be calculated efficiently by Fourier methods. On the other hand allowing arbitrary affine processes destroys the nonnegativity of forward interest rates, a central property of affine LIBOR models. We propose a modification, that preserves the nonnegativity of forward interest rates without the restriction to nonnegative affine processes.

\section{The modified affine LIBOR model} \label{sec:modaffinemodel} 

On the filtered probability space $(\Omega, \mathcal{A}, (\mathcal{F}_t)_{0 \leq t \leq T}, \QM^T)$ consider an analytic one-dimensional affine process $X$ with a fixed starting value $x_0$, i.e. the set $\mathcal{V}$ defined in \eqref{eq:momset} contains $0$ in the interior. For $u \in \mathcal{V}$ with $-u \in \mathcal{V}$ consider the martingales $M^u$,
\begin{equation} \label{eq:coshmartingales}
M_t^u :=  \EV[\QM^T]{\cosh(u X_T) \vert \F_t}  = \frac{1}{2} \left(\e^{\phi_{T-t}(u) + \psi_{T-t}(u) X_t} + \e^{\phi_{T-t}(-u) + \psi_{T-t}(-u) X_t} \right).
\end{equation}
By the symmetry of the cosinus hyperbolicus $M^u = M^{-u}$, hence one may restrict $u$ to be nonnegative. 
For the given tenor structure $0 < T_1  < \cdots < T_N \leq T_{N+1} = T$ and the market setup of section \ref{sec:marketmodels} define the normalized bond prices for $k=1, \dots, N$ and $t\leq T_k$ as
\begin{equation*}
\begin{aligned}
\frac{ P(t,T_k)}{P(t,T)} & :=  M_{t}^{u_k}, \qquad u_k \in \{v \in \mathcal{V}: v \geq 0, -v \in \mathcal{V} \}.
\end{aligned}
\end{equation*} 
With $M_t^{u_k}$ being a $\QM^T$-martingale the model is arbitrage-free. For every $x \in \R$ the function $u \mapsto \cosh(u x)$ is increasing in $u \in \R_{\geq 0}$ and satisfies $\cosh(ux) \geq 1$ so that if
$$u_1 \geq u_2 \geq \dots \geq u_N \geq 0,$$
equation \eqref{eq:bondmonotonicity} holds and forward interest rates
\begin{equation*} F^k(t) = \frac{1}{\Delta_k} \left( \frac{M_t^{u_{k-1}}}{M_t^{u_k}} - 1 \right), \quad 0 \leq t \leq T_{k-1}. \end{equation*}
are nonnegative for all $t$. 
To fit initial market data one has to choose the sequence $(u_k)$ so that $M_0^{u_k} = {P(0,T_k)}/{P(0,T)}. $
The following lemma gives the condition for the affine process $X$ under which a given initial term structure can be reproduced and shows that the $u_k$ are uniquely determined.

\begin{lemma} \label{thm:interestfit}
If $$P(0,T_1) / P(0,T) < \sup_{u \in \mathcal{V}: -u \in \mathcal{V}} \EV[\QM^T]{\cosh(u X_T) \vert \F_0},$$ then the model can fit any term structure of nonnegative forward interest rates. Additionally there exists a unique decreasing sequence $u_1 \geq \cdots \geq u_N$, such that  $${ P(0,T_k)}/{P(0,T)} = \EV[\QM^T]{\cosh(u_k X_T) \vert \F_0} = M_0^{u_k}.$$
If forward interest rates are strictly positive, the sequence is strictly decreasing.
\end{lemma}
\begin{proof}
$m(u) = \EV[\QM^T]{\cosh( u X_T) \vert \F_0}$ is a continuous function which is strictly increasing for $ u \geq 0$. By the assumption of the theorem there exists $\overline{u} >0$ with $m(\overline{u})
 > P(0,T_1)/P(0,T)$. Furthermore $
m(0) = 1$, which proves the lemma.
\end{proof}
 \begin{remark}
Generalizing this approach to a $d$-dimensional driving process is possible by setting $$M_t^u = \EV{\prod_{l=1}^d \cosh \left(u^{(l)} X_T^{(l)} \right) \Big \vert \F_t}, \quad u =(u^{(1)},\dots,u^{(d)}) \geq 0.$$
In this case it is guaranteed that $M_t^{u} \geq M_t^w$ for $u \geq w$, which guarantees the nonnegativity of forward interest rates.  However, the option pricing formulas in the following sections do not generalize.
\end{remark}

As in the affine LIBOR model for a monotonically decreasing sequence $(u_k)$ forward interest rates are not only nonnegative, but bounded below by strictly positive time-dependent constants (the bounds can be calculated numerically). This is not a big issue if these bounds are close to zero, but has to be checked during the calibration process. 

In the modified affine LIBOR model the change of measure to the $T_k$-forward measure $\QM^{T_k}$ is given by
\begin{equation} \frac{\dd{\QM^{T_k}}}{\dd{\QM^{T}}}  = \frac{P(0,T)}{P(0,T_k)}  M_{T_k}^{u_k} = \frac{M_{T_k}^{u_k}}{M_0^{u_k}}. \label{eq:coshmeasurechange} \end{equation}
Here $M_t^{u_k}$ is a sum of exponentials of $X_t$, while in the affine LIBOR model the corresponding term is a single exponential. This means that contrary to the affine LIBOR model the process $X$ is not an inhomogeneous affine process under $\QM^{T_k}$ and it is not possible to calculate the moment generating function of the logarithm of foward bond prices under $\QM^{T_k}$. Nevertheless it is possible to get analytical formulas for the prices of caplets and swaptions. 

\subsection{Option pricing} \label{sec:option pricing}
The derivation of the pricing formulas for caplets and swaptions is based on a method first applied in \citet{JA89}. First caplets are dealt with, swaptions follow afterwards\footnote{Actually caplet prices coincide with prices of swaptions with only one underlying period. The difference between those two derivatives is the payoff time.}.
Note that if $u_k = u_{k-1}$ the corresponding forward interest rate $F^{k}$ always stays zero. To exclude such pathological examples assume that the sequence $(u_k)$ is strictly decreasing.
In this section random variables are often viewed as functions of the value of the driving process $X$. Specifically consider the functions $M_t^u: \R \rightarrow \R$, 
\begin{equation} \label{eq:coshmartingalefunction}
x \mapsto M_t^u(x) := \frac{1}{2} \left(\e^{\phi_{T-t}(u) + \psi_{T-t}(u) x} + \e^{\phi_{T-t}(-u) + \psi_{T-t}(-u) x} \right).\end{equation}
The time $t$ value of martingale $M^u$ in \eqref{eq:coshmartingales} is then $M_t^u = M_t^u(X_t)$. In the rest of the paper $M_t^u$ will denote both, the function and the value of the stochastic processes, where the correct interpretation should be clear from context. 

The payoff of a caplet for the $(k+1)^\text{th}$ forward rate $F^{k+1}(T_{k})$ with strike $K$ is
$$ \Delta_{k+1} \left(F^{k+1}(T_k)- K \right)_+ = \left(\frac{1}{P(T_k,T_{k+1})} - \tilde{K} \right)_+ = \left( \frac{M_{T_k}^{u_{k}}}{M_{T_k}^{u_{k+1}}} - \tilde{K} \right)_+, $$
where $\tilde{K} = 1 + \Delta_{k+1} K$. Since this payoff has to be paid at time $T_{k+1}$ 
the price of the caplet and the corresponding floorlet is
\begin{align*}
\mathrm{Cpl}(t,T_k,T_{k+1},K) &=  P(t,T_{k+1}) \EV[\QM^{T_{k+1}}]{\left( \frac{M_{T_k}^{u_{k}}}{M_{T_k}^{u_{k+1}}} - \tilde{K} \right)_+ \Big \vert \F_t}, \notag \\
\mathrm{Flt}(t,T_k,T_{k+1},K) & =  P(t,T_{k+1}) \EV[\QM^{T_{k+1}}]{\left(\tilde{K} -  \frac{M_{T_k}^{u_{k}}}{M_{T_k}^{u_{k+1}}}  \right)_+ \Big \vert \F_t}. 
\end{align*}
Since price processes are martingales, the put/call parity holds and prices of caplets follow from floorlets and vice versa. Because Fourier analysis is easier for floorlets, where the payoff is bounded, formulas are derived for floorlets.

Since the moment generating function of $\ln({{M_{T_k}^{u_{k}}}/{M_{T_k}^{u_{k+1}}}})$ is unknown, Fourier methods are not directly applicable. However, 
the function $x \mapsto {M_{T_k}^{u_{k}}(x)}/{M_{T_k}^{u_{k+1}}(x)}$ has a unique minimum and is monotonically increasing moving away from this minimum. Using this one can get rid of the positive part and use Fourier inversion to calculate the above expectations. 
The above mentioned monotonicity is very fortunate and follows from a close interplay between the monotonicity of the sequence $(u_k)$ and the function $\psi$ with properties of the cosinus hyperbolicus. Details are laid out in the proof of the following lemma, which can be found in the appendix. 
\begin{lemma} \label{lem:monoton} For $i=1, \dots, n$ let $u_0 \geq u_i \geq 0$, where for at least one $i$ $u_0 > u_i$. Let $c_i > 0$ be positive constants. Define a function $g: \R \rightarrow \R$ by 
\begin{equation}
g(x) := \sum_{i=1}^n c_i \frac{M_t^{u_i}(x)}{M_t^{u_0}(x)} \p \label{eq:monoton}
\end{equation}
Then $g$ has a unique maximum at some point $\xi \in \R$ and and is strictly monotonically decreasing to $0$ on the left and right side of $\xi$.
\end{lemma}

For floorlet valuation this lemma is not directly applicable as $u_k > u_{k+1}$, which is the wrong inequality. However, there is only one summand and the lemma can be applied to the inverse ${M_{T_k}^{u_{k+1}}(x)}/{M_{T_k}^{u_{k}}(x)}$. 
It follows that ${M_{T_k}^{u_{k}}(x)}/{M_{T_k}^{u_{k+1}}(x)}$ has a unique minimum at some point $\xi$ and is increasing to infinity to the left and right. 
Hence it is possible to write
\begin{equation} \label{eq:fltpayoff}
\left(\tilde{K} -  \frac{M_{T_k}^{u_{k}}(x)}{M_{T_k}^{u_{k+1}}(x)}  \right)_+ = \left(\tilde{K} -  \frac{M_{T_k}^{u_{k}}(x)}{M_{T_k}^{u_{k+1}}(x)}  \right) \1{\kappa_1 < x <  \kappa_2},
\end{equation}
where $\kappa_1$ and $\kappa_2$ are two uniquely determined constants satisfying $\kappa_1 \leq \xi \leq \kappa_2$. 
If $\kappa_1 = \xi = \kappa_2$ the payoff is 
zero, which corresponds to ${M_{T_k}^{u_{k}}}/{M_{T_k}^{u_{k+1}}}  > \tilde{K}$. This happens if the forward interest rate is bounded from below by K, 
which only happens for very low strikes $K$. Inserting \eqref{eq:fltpayoff} into the price of a floorlet it follows by a change of measure that
\begin{align}
\mathrm{Flt}(t,T_k,T_{k+1},K) 
& =  P(t,T_{k+1}) \EV[\QM^{T_{k+1}}]{\left(\tilde{K} -  \frac{M_{T_k}^{u_{k}}}{M_{T_k}^{u_{k+1}}}  \right)  \1{\kappa_1 < X_{T_k} <  \kappa_2} \Big \vert \F_t} \notag \\
&=  P(t,T) \EV[\QM^{T}]{\left(\tilde{K} M_{T_k}^{u_{k+1}} -  M_{T_k}^{u_{k}}   \right)  \1{\kappa_1 < X_{T_k} <  \kappa_2} \Big \vert \F_t}. \label{eq:floorlet}
\end{align}
$\tilde{K} M_{T_k}^{u_{k+1}} -  M_{T_k}^{u_{k}}$ is the sum of exponentials of the random variable $X_{T_k}$.
The expectation in \eqref{eq:floorlet} is calculated under the measure $\QM^T$, where the conditional moment generating function
$$\mathcal{M}_{X_t \vert X_s}(z) := \EV[\QM^T]{\e^{z X_t}\vert \F_s} =  \EV[\QM^T]{\e^{z X_t}\vert X_s} = \ex{\phi_{t-s}(z) + \psi_{t-s}(z) X_s}$$
is known for $z \in \mathcal{V}$. Hence the expectation in \eqref{eq:floorlet} can be calculated via Fourier inversion. The Fourier inversion formula for terms of the above form is stated in Lemma \ref{lem:fouriertransform}, the proof of which is given in the appendix.

\begin{lemma}
\label{lem:fouriertransform}
Assume that the function $f: \R \rightarrow \R$ has the representation
$$f(x) = \sum_k C_k \e^{v_k x}\1{\kappa_1 < x < \kappa_2}, \qquad \lim_{x \downarrow \kappa_1} f(x) = \lim_{x \uparrow \kappa_2} f(x) = 0,$$
where the summation is over a finite index set and the $C_k$ and $v_k$ are real constants.
Then for $R \in \mathcal{V} \cap \R$ 
the Fourier inversion formula  
\begin{equation*}
\EV{f(X_t)\vert \F_s} = \frac{1}{\pi} \int_{0}^{\infty} \mathrm{Re} \left(  { \mathcal{M}_{X_t \vert X_s}(\i u + R)} \hat{f}(u- \i R) \right)  \dd{u}
\end{equation*}
holds, where $\hat{f}$ is the analytic Fouier transform given by
\begin{equation}
\hat{f}(z)  = \frac{1}{\i z} \sum_k  {\frac{C_k v_k }{v_k - \i z}  { \left(\e^{(v_k - \i z) \kappa_2} - \e^{(v_k - \i z) \kappa_1 } \right) } }, \qquad z \neq 0, z \neq -\i v_k.
\end{equation}
\end{lemma}

To calculate the price of a floorlet in \eqref{eq:floorlet} apply Lemma \ref{lem:fouriertransform} to
$f_{k+1}^{K}(X_{T_{k}})$ with
\begin{equation}
f_{k+1}^{K}(x) := \left(\tilde{K} M_{T_k}^{u_{k+1}}(x) -  M_{T_k}^{u_{k}}(x)   \right)  \1{\kappa_1 < x <  \kappa_2}.  \label{eq:ffloorlet}
\end{equation}  
Its Fourier transform is
\begin{equation} \label{eq:fhatfloorlet}
\hat{f}^{K}_{k+1}(z) = \frac{1}{\i z} \left( (1 + \Delta_{k+1} K)  h_{\kappa_1,\kappa_2}^{T_{k}} (-\i z,u_{k+1}) -  h_{\kappa_1,\kappa_2}^{T_{k}}(-\i z,u_{k}) \right)
\end{equation}
with
\begin{equation}  \label{eq:hfunction}
\begin{split}
 h_{\kappa_1,\kappa_2}^t(z,u) & := \e^{\phi_{T-t}(u)} \frac{\psi_{T-t}(u)}{2 ( z + \psi_{T-t}(u))} \left(\e^{(z + \psi_{T-t}(u)) \kappa_2} - \e^{( z + \psi_{T-t}(u)) \kappa_1 } \right) \\
& + \e^{\phi_{T-t}(-u)} \frac{\psi_{T-t}(-u)}{2 (z + \psi_{T-t}(-u))} \left(\e^{( z + \psi_{T-t}(-u)) \kappa_2} - \e^{( z + \psi_{T-t}(-u)) \kappa_1 } \right).
\end{split}
\end{equation}

The case of swaptions is similar. Consider a swap which is part of the tenor structure. That is, consider $1 \leq \alpha < \beta \leq N$ and the according interest rate swap with forward swap rate
$$  S_{\alpha,\beta}(t)  = \frac{P(t,T_\alpha) - P(t,T_\beta)}{\sum_{k=\alpha+1}^\beta \Delta_k P(t,T_k)}, \qquad  \Delta_k = T_{k} - T_{k-1}. $$
The payoff of a put swaption on the above swap with strike $K$ is then 
\begin{align*}
\sum_{k=\alpha+1}^\beta P(T_\alpha,T_k) \Delta_k \left(K - S_{\alpha,\beta}(T_\alpha) \right)_+ & = \left(P(T_\alpha,T_\beta) + K \sum_{k=\alpha+1}^\beta \Delta_k P(T_\alpha,T_k) - 1\right)_+ \\
& = \left(\frac{M_{T_\alpha}^{u_\beta}}{M_{T_\alpha}^{u_\alpha}} +  \sum_{k=\alpha+1}^\beta K \Delta_k \frac{M_{T_\alpha}^{u_k}}{M_{T_\alpha}^{u_\alpha}} - 1 \right)_+ .
\end{align*}
Since the function ${M_{T_\alpha}^{u_\beta}(x)}/{M_{T_\alpha}^{u_\alpha}(x)} +  \sum_{k=\alpha+1}^\beta K \Delta_k {M_{T_\alpha}^{u_k}(x)}/{M_{T_\alpha}^{u_\alpha}(x)} $ is of the form of Lemma \ref{lem:monoton}, it has a unique maximum $\xi$ and one can find constants\footnotemark  $\ \kappa_1 \leq \xi \leq \kappa_2$ such that 
after a change of measure the value of a put swaption is 

\footnotetext{
As in the floorlet case if $\kappa_1 = \kappa_2$ then the forward swap rate is always larger than the strike. Note that $S_{\alpha,\beta}(t)$ can also be written as 
$S_{\alpha,\beta}(t) = \sum_{k=\alpha+1}^\beta w_k(t) F^k(t)$
with $w_k > 0$ (see e.g. \citet{BM06}). It follows that if forward interest rates are bounded below by positive constants the same will be true for forward swap rates. This bound is then at most an average of the corresponding forward interest rates bounds and is therefore of the same order of magnitude, which for a meaningful model will be small enough.
}

\begin{equation*}
\mathrm{Put Swaption}(t,T_\alpha,T_\beta,K) = P(t,T) \EV[\QM^T]{f^K_{\alpha,\beta}(X_{T_\alpha}) \Big \vert \F_t},
\end{equation*}
where 
\begin{equation} f^K_{\alpha,\beta}(x) = \left( {M_{T_\alpha}^{u_\beta}(x)}  -  M_{T_\alpha}^{u_\alpha}(x) +  \sum_{k=\alpha+1}^\beta K \Delta_k {M_{T_\alpha}^{u_k}(x)} \right) \1{\kappa_1 < x < \kappa_2}.
\label{eq:fputswaption}
\end{equation}
Again this of the form in Lemma \ref{lem:fouriertransform} and in this case
\begin{equation} \label{eq:fhatputswaption}
\begin{aligned}
\hat{f}^K_{\alpha,\beta}(z) = 
\frac{1}{\i z} \Big( &   h_{\kappa_1,\kappa_2}^{T_\alpha}(-\i z,u_\beta) -   h_{\kappa_1,\kappa_2}^{T_\alpha}(- \i z ,u_\alpha) + {K} \sum_{k=\alpha+1}^\beta \Delta_k h_{\kappa_1,\kappa_2}^{T_\alpha}(-\i z,u_k) \Big),
\end{aligned}
\end{equation}
where $h^t_{\kappa_1,\kappa_2}(z,u)$ is defined in \eqref{eq:hfunction}.
The pricing formulas are summarized in the following theorem.
\begin{theorem}
Let $R \in \mathcal{V} \cap \R$. In the modified affine LIBOR model prices of a forward interest rate put and a put swaption are
\begin{equation}
\Flt(t,T_k,T_{k+1},K) = \frac{P(t,T)}{\pi} \int_{0}^{\infty} \mathrm{Re} \left(  { \mathcal{M}_{X_{T_k}\vert X_t}(R + \i u)}
\hat{f}^{K}_{k+1}(u - \i R) \right) \dd{u}, 
\end{equation}
\begin{equation}
\mathrm{Put Swaption}(t,T_\alpha,T_\beta,K) = \frac{P(t,T)}{\pi} \int_{0}^{\infty} \mathrm{Re} \left(  { \mathcal{M}_{X_{T_\alpha}\vert X_t}(R + \i u)}
\hat{f}^K_{\alpha,\beta}(u - \i  R) \right) \dd{u},
\end{equation}

The Fourier transforms
$\hat{f}^{K}_{k+1}$ repectively $\hat{f}^K_{\alpha,\beta}$ are given in \eqref{eq:fhatfloorlet} respectively \eqref{eq:fhatputswaption} for $R \notin \{0,u_k, u_{k+1}\}$ respectively $R \notin \{0,u_\alpha, \dots, u_\beta \}$.
\end{theorem}

In order to calculate $\hat{f}^{K}_{i}$ respectively $\hat{f}^K_{\alpha,\beta}$ one has to find the roots $\kappa_1, \kappa_2$ of the functions
\begin{align} 
g^K_{k}(x) & := \tilde{K} -  \frac{M_{T_k}^{u_{k}}(x)}{M_{T_k}^{u_{k+1}}(x)},  \label{eq:fltfunc} \\
{g}^K_{\alpha,\beta}(x) & := \frac{M_{T_\alpha}^{u_\beta}(x)}{M_{T_\alpha}^{u_\alpha}(x)} +  \sum_{k=\alpha+1}^\beta K \Delta_k \frac{M_{T_\alpha}^{u_k}(x)}{M_{T_\alpha}^{u_\alpha}(x)} - 1 .\label{eq:swaptionfunc}
\end{align}
By Lemma \ref{lem:monoton} this amounts to finding the roots of a function which has a single optimum and is monotonic when moving away from this optimum.
Numerical determination of the roots of such well-behaved one-dimensional functions poses no problem. 
Having determined those bounds valuation reduces to a one-dimensional integration of a function that is falling at least like $1/x^2$ (depending on the moment generating function of the affine process), so also numerical integration is feasible.
Note that besides caps, floors and swaptions, options like digital options or Asset-or-Nothing options can be calculated in a similar manner.

\subsection{Examples} \label{sec:examples}
The first part of this section looks at the benchmark case of a Brownian motion, where everything can also be calculated in closed form. Afterwards Ornstein-Uhlenbeck processes are discussed. The section concludes with examples of possible volatility surfaces.
\subsubsection*{Brownian motion}
Choose $X_t = B_t$, a standard Brownian motion starting in 0. 
The conditional moment generating function is
$$\mathcal{M}_{B_T \vert B_t}(u) = \EV{\e^{u B_T} \vert \F_t} = \ex{u B_t + \frac{u^2}{2} (T-t)}.$$
Hence this is an affine process with $\phi_t(u) = \frac{u^2}{2} t$ and $\psi_t(u) = u$. 
Consider the time $0$ price of a floorlet as given in \eqref{eq:floorlet} with $t=0$.
Since $M_t^u(-x) = M_t^u(x)$ one finds that in this case  
$\kappa_2 = \kappa$ and $\kappa_1 = - \kappa$, where $\kappa$ is the unique positive root of \eqref{eq:fltfunc} if $g^{K}_{k+1}(0) < 0$ and $\kappa=0$ otherwise. 
By \eqref{eq:floorlet} the floorlet price $\Flt(0,T_k,T_{k+1},K)$ is
\begin{align*}P(0,T) \EV[\QM^{T}]{ \left( {\tilde{K}} \e^{\frac{u_{k+1}^2}{2} (T-T_k)}\cosh(u_{k+1} B_{T_k}) -  \e^{\frac{u_{k}^2}{2} (T-T_k)} \cosh(u_{k} B_{T_k})\right)  \1{\vert B_{T_k} \vert \leq \kappa} }. 
\end{align*}
By the symmetry of a Brownian motion starting in $0$
\begin{align*}  \EV{\cosh(z B_t) \1{\vert B_t \vert \leq \kappa}} & = \EV{\e^{z B_t} \1{\vert B_t \vert \leq \kappa}} = \EV{\e^{-z B_t} \1{\vert B_t \vert \leq \kappa}} \\
&= \e^{\frac{1}{2} t z^2} \left( \Phi \Big(\frac{\kappa}{\sqrt{t}}- z \sqrt{t} \Big) - \Phi \Big (-\frac{\kappa}{\sqrt{t}} - z \sqrt{t} \Big) \right) ,
\end{align*}
where $\Phi$ denotes the cumulative distribution function of a standard normal distributed random variable. Hence
\begin{align*}
\mathrm{Flt}(0,T_k,T_{k+1},K)  =  {\tilde{K}} & P(0,T) \e^{u_{k+1}^2  \frac{T}{2} } \left( \Phi \Big(\frac{\kappa}{\sqrt{T_k}}- u_{k+1} \sqrt{T_k} \Big) - \Phi \Big(- \frac{\kappa}{\sqrt{T_k}}-u_{k+1} \sqrt{T_k}\Big)  \right) \\
- & P(0,T) \e^{u_k^2  \frac{T}{2} } \left( \Phi \Big(\frac{\kappa}{\sqrt{T_k}}-u_k \sqrt{T_k}\Big) - \Phi \Big(- \frac{\kappa}{\sqrt{T_k}}-u_k \sqrt{T_k}\Big)  \right).
\end{align*}
Slightly more complicated formulas exist when $B$ is replaced with a Brownian motion with constant drift and volatility and a starting value different from $0$. 

Swaptions can be treated the same way. Let $\kappa$ be the unique positive root of \eqref{eq:swaptionfunc} if $g^K_{\alpha,\beta}(0) > 0$ and $\kappa=0$ otherwise. Then
\begin{align*}
\mathrm{Put Swaption}&(0,T_\alpha,T_\beta,K) =  P(0,T)   \e^{u_\beta^2  \frac{T}{2} } \left( \Phi \Big( \frac{\kappa}{\sqrt{T_\alpha}}- u_\beta \sqrt{T_\alpha}\Big) - \Phi \Big(- \frac{\kappa}{\sqrt{T_\alpha}}-u_\beta \sqrt{T_\alpha} \Big)  \right) \\
& - P(0,T) \e^{u_\alpha^2  \frac{T}{2} } \left( \Phi \Big(\frac{\kappa}{\sqrt{T_\alpha}}-u_\alpha \sqrt{T_\alpha} \Big) - \Phi \Big(-\frac{\kappa}{\sqrt{T_\alpha}} -u_\alpha \sqrt{T_\alpha} \Big)   \right) \\
& + P(0,T)  \sum_{k=\alpha+1}^\beta K \Delta_k \e^{u_k^2  \frac{T}{2} } \left( \Phi \Big( \frac{\kappa}{\sqrt{T_\alpha}}- u_k \sqrt{T_\alpha}\Big) - \Phi \Big(- \frac{\kappa}{\sqrt{T_\alpha}}- u_k \sqrt{T_\alpha} \Big)  \right).
\end{align*}

\subsubsection*{Ornstein-Uhlenbeck (OU) processes} \label{sec:OUprocess}
The OU process $X$ generated by a Lévy process $L$ is defined as the unique strong solution of (see \citet{SA99}, section 17)
\begin{equation} \dd{X_t} = - \lambda X_t \dd{t} + \dd{L_t}, \quad X_0 = x_0.  \label{eq:OUprocess} \end{equation}
Then $ Y_t := \e^{\lambda t} X_t 
= x_0 + \int_0^t \e^{\lambda s} L_s \dd{s}. $
Using the key formula of \citet{ER99} it follows that
$$ \EV{\e^{u X_t}} = \EV{\ex{\e^{-\lambda t} u Y_t}} = \ex{\e^{-\lambda t} u x_0 + \int_0^t \kappa(\e^{-\lambda s} u) \dd{s}},$$
where $\kappa(u) = \logn{\EV{L_1}}$ is the cumulant generating function of the Lévy process $L$.
Hence this process is affine with 
\begin{equation} \psi_t(u) = \e^{-\lambda t} u \quad \text{ and } \quad \phi_t(u) =  \int_0^t \kappa(\e^{-\lambda s} u) \dd{s} = \frac{1}{\lambda} \int_{\e^{-\lambda t}}^1 \frac{\kappa(v u)}{v} \dd{v}. \label{eq:phiOU} \end{equation}
By Corollary 2.10 in \citet{DFS03} every affine process with state space $\R$ is in fact an OU process. Hence in the context of affine processes defined on the real line OU processes are the right class to consider.  
For application it should be possible to calculate the integral in \eqref{eq:phiOU} analytically. Two examples where this is possible are presented below.

\begin{remark}
If $L$ is a martingale, the process in \eqref{eq:OUprocess} is mean-reverting to zero, however shifting the mean to $\theta$ is easily done by using $Z_t = \theta+X_t.$ Then $\dd{Z_t} = \lambda (\theta-Z_t) \dd{t} + \dd{L_t}$ and
$$\EV{\e^{u Z_t} } = \ex{(\phi_t(u) + \theta u  (1-\e^{-\lambda t})) + \psi_t(u) Z_0}.$$ 
Hence $Z$ is again affine with $\psi_t^\theta(u) = \psi_t(u)$ and $\phi_t^\theta(u) = \phi_t(u) + \theta u  (1-\e^{-\lambda t})$.
Note that this process is then generated by the Lévy process $\tilde{L}_t = L_t + \theta \lambda t$, i.e. the original Lévy process plus an additional drift of $\theta \lambda$.
\end{remark}

The first example is the classical OU process generated by a Brownian motion $ \sigma B$, where $\kappa(u) = \frac{1}{2} \sigma^2 u^2$. This process is described by
$$\dd{X_t} =  - \lambda  X_t \dd{t} + \sigma \dd{B_t},  \quad X_0 = x. $$ The integral in \eqref{eq:phiOU} is 
\begin{align} \label{eq:phicont}
\phi_t(u) = \frac{1}{\lambda} \int_{\e^{-\lambda t}}^1 \frac{ \kappa(v u)}{v} \dd{v} 
 =  \frac{\sigma^2 u^2}{4 \lambda} (1-\e^{-2 \lambda t}).
\end{align}

%
%
%
%
%
%

With Brownian motion describing the continuous part of Lévy processes., for the second example we consider a pure jump process, namely a Double $\Gamma$-OU process. $\Gamma$-OU processes are generated by a compound Poisson process with jump intensity $\lambda \beta$ ($\lambda$ being the same as in \eqref{eq:OUprocess}) and exponentially distributed jumps with expectation value $\alpha$. The limit distribution of this process is a $\Gamma$-distribution, which gives the process its name. 
As the generating compound Poisson process is strictly increasing, the generated $\Gamma$-OU process is a subordinator and stays above 0. In order to find an OU process with values in $\R$ consider the difference of two independent compound $\Gamma$-OU processes $L^+, L^-$ with parameters $\alpha^+, \beta^+, \alpha^-, \beta^-$ and set $\lambda^+ = \lambda \beta^+, \lambda^- =  \lambda \beta^-$. Then $L = L^+ - L^-$ is a compound Poisson process, where positive jumps with expected jump size $\frac{1}{\alpha^+}$ are arriving at rate $\lambda^+$, while negative jumps with expected jump size $\frac{1}{\alpha^-}$ are arriving at rate $\lambda^-$. 

The cumulant generating function of a compound Poisson process with exponential jumps is $\frac{\lambda \beta u}{\alpha -u}$, which is defined for $u < \alpha$. Hence for $u \in (-\alpha^-,\alpha^+)$ the moment generating function of the combined process $L$ is
\begin{align*}
\EV{\e^{u L_1}} & 
= \EV{\e^{u L^+_1 }} \EV{\e^{ - u L^-_1}} 
= \ex{\lambda \frac{(\beta^+ + \beta^-) u^2 + (\beta^+ \alpha^- - \beta^- \alpha^+) u}{(\alpha^+ -u)(\alpha^- +u)}}.
\end{align*}
Inserting this into \eqref{eq:phiOU} straightforward calculations show that the function $\phi$ for the resulting OU process is given by
\begin{equation} \label{eq:phijump}
\begin{aligned}
\phi_t(u) & = \frac{\beta^+ + \beta^-}{2} \logn{\frac{(\alpha^+-\e^{-\lambda t}u)(\alpha^-+\e^{-\lambda t}u)}{(\alpha^+-u)(\alpha^-+u)}} \\
& + \frac{\beta^+ - \beta^-}{2 } \logn{\frac{(\alpha^+-\e^{-\lambda t}u)(\alpha^-+u)}{(\alpha^+-u)(\alpha^-+\e^{-\lambda t}u)}}.
\end{aligned}
\end{equation}

It is also possible to combine the two approaches by considering an OU process generated by a Lévy process which is the difference of two compound Poisson processes plus a Brownian motion, all of which are independent. The resulting $\phi$ then follows by adding up the two functions \eqref{eq:phicont} and \eqref{eq:phijump} and for this process $\mathcal{V} = \{u \in \C: - \alpha^- < \mathrm{Re}(u)  < \alpha^+ \}$. By the previous remark it is also possible to shift this process by $\theta$. Such OU processes are used in the following numerical examples.

\subsubsection*{Volatility surfaces}
With the OU process of the previous section it is possible to generate volatility smiles as well as volatility skews. For illustration we consider a term structure with constant interest rates of $3.5\%$. The tenor structure and therefore the forward interest rates are based on half year intervals. Implied volatilities are then calculated for caplets with maturities over a 5-year period and strikes ranging from $0.02$ to $0.07$. Figure \ref{fig:OUskew} shows a skewed volatility surface while figure \ref{fig:OUsmile} shows a very pronounced smile, both of which are generated by an $OU$ process of the just introduced type. 
As mentioned in the previous chapters forward interest rates in this type of model will be bounded from below. The bounds in these examples are at $1\%$ for the forward interest rate expiring after half a year and decrease to basically $0\%$ for the forward interest rate which expires in 5 years. Hence they are well within reasonable boundaries. 
\begin{figure}
\centering
\includegraphics[width=0.7 \textwidth, trim=0cm 2cm 0cm 2.5cm,clip=true]{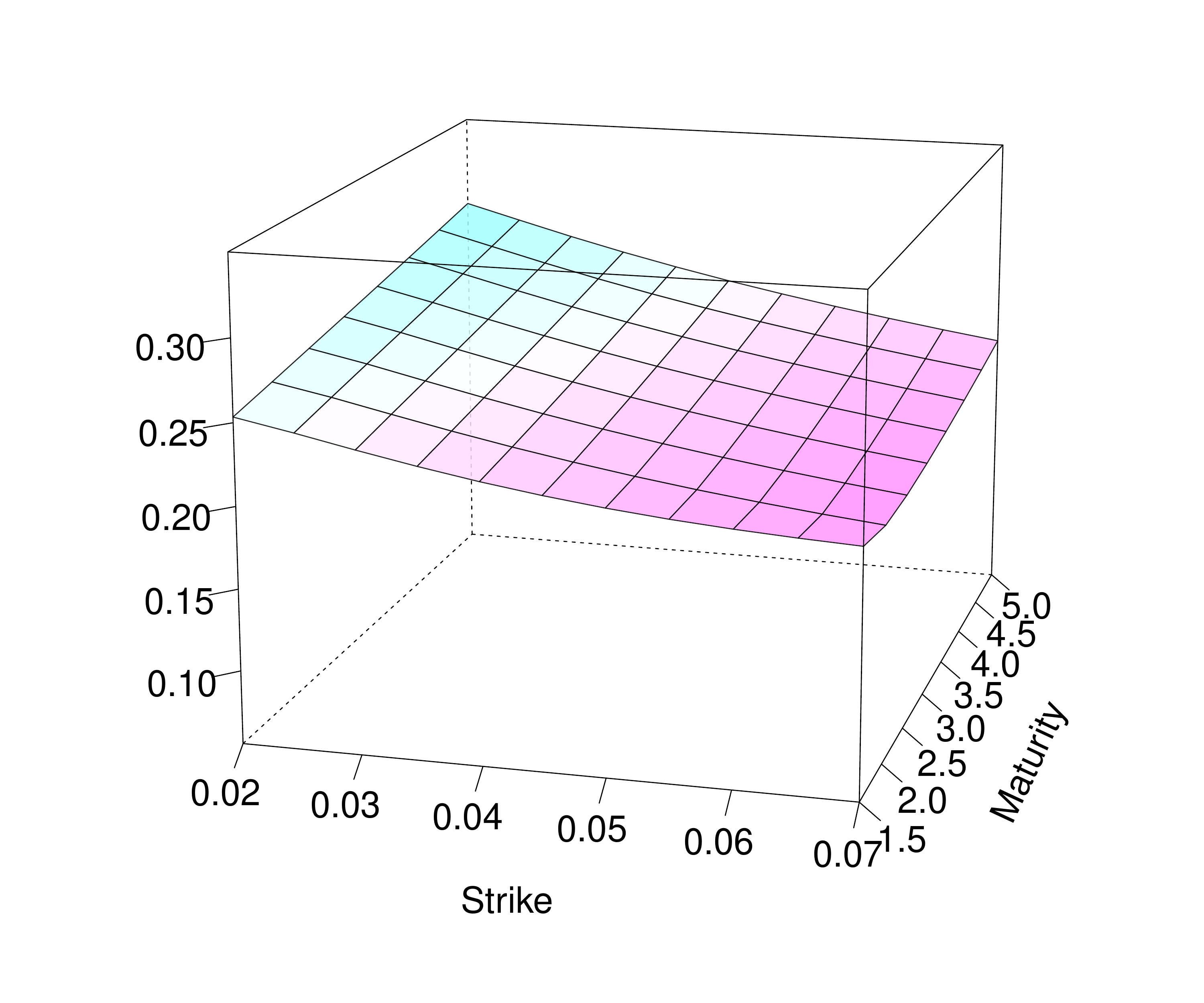}
\captionsetup{singlelinecheck=false, margin = 1.5 cm}
\caption[OU skew]{Implied volatility skew of caplets generated by an OU process with parameters $\lambda = 0.02, \alpha^+ = 12, \alpha^- =10, \beta^+=50, \beta^- = 5, \sigma=0.3, \theta=0.5, x=0.7$ and $T = 10$.}
\label{fig:OUskew}
\end{figure}
\begin{figure}
\centering
\includegraphics[width=0.7 \textwidth,trim=0cm 2cm 0cm 2.5cm,clip=true]{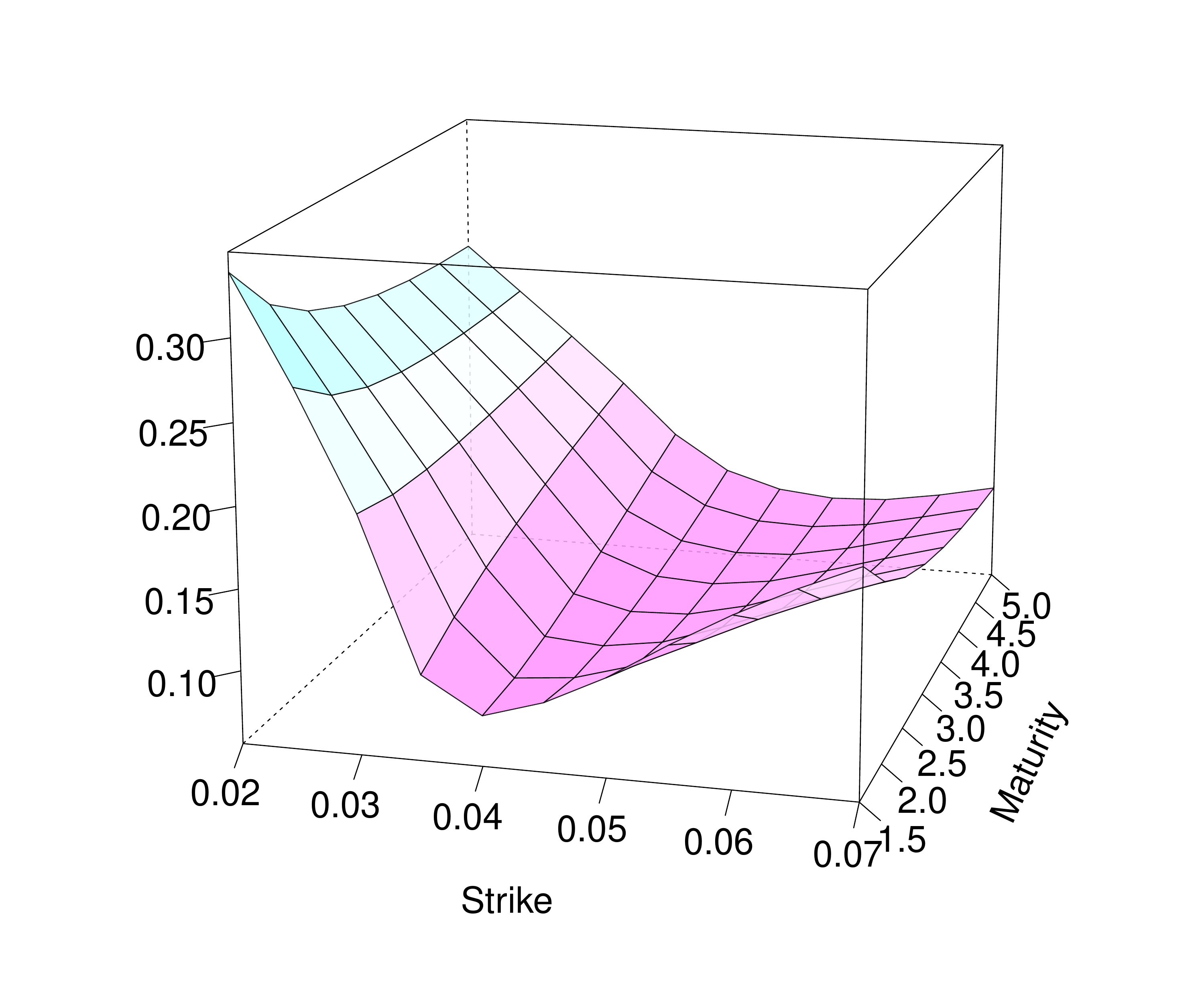}
\captionsetup{singlelinecheck=false, margin = 1.5 cm}
\caption[OU smile]{Implied volatility smile of caplets generated by an OU process with parameters $\lambda = 0.02, \alpha^+ = 50, \alpha^- =5, \beta^+=50, \beta^- = 10, \sigma=0, \theta=0, x=1$ and $T = 10$.}
\label{fig:OUsmile}
\end{figure}
For completeness an example of at-the-money implied volatilities for swaptions with maturities and underlying swap rates ranging from 2 to 7 years is displayed in figure \ref{fig:swaptionsurface}. 

\begin{figure}
\centering
\includegraphics[width=0.7 \textwidth]{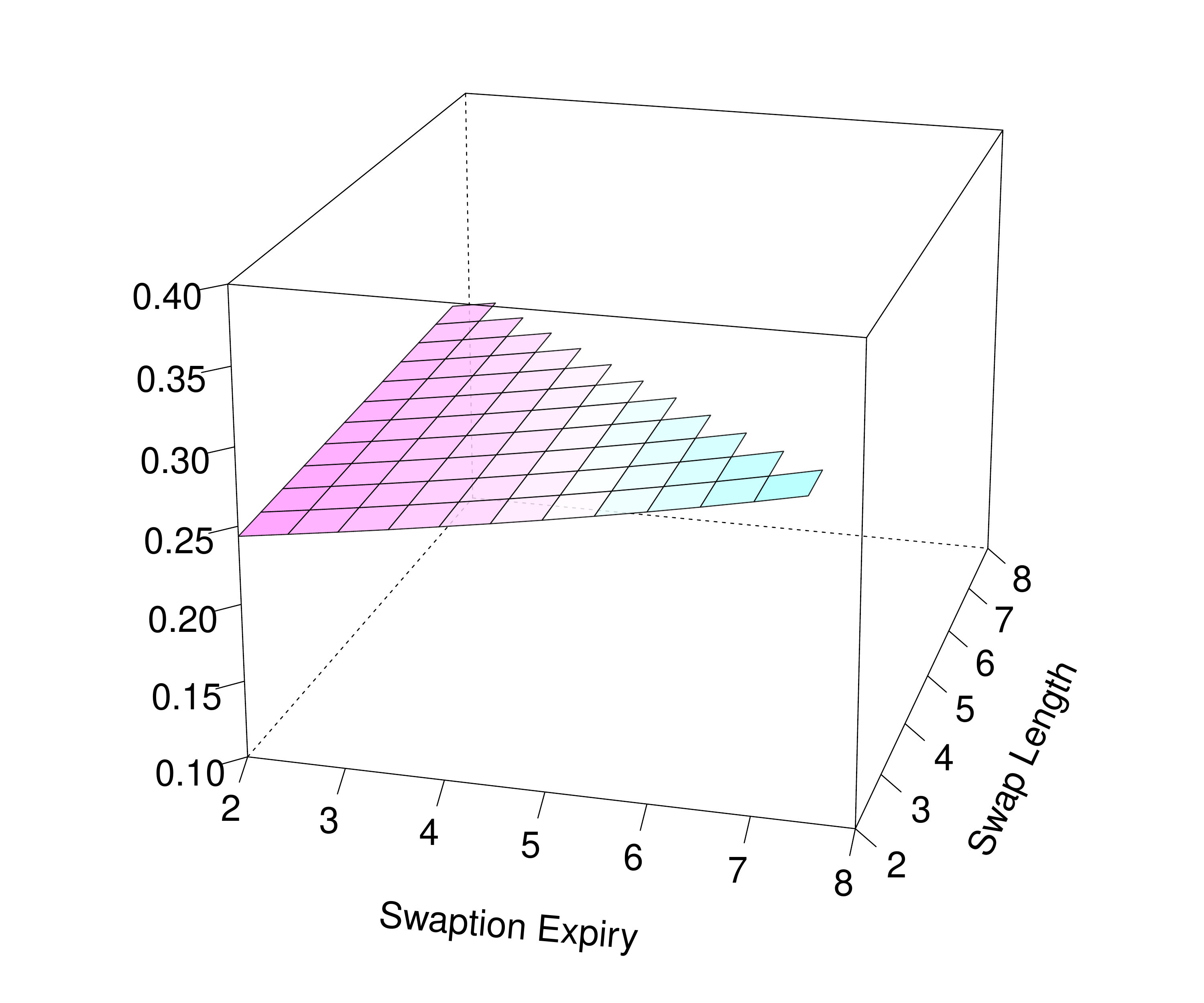}
\captionsetup{singlelinecheck=false, margin = 1.5 cm}
\caption[Swaption Implied volatilites]{Swaption implied volatilites generated by an OU process with parameters $\lambda = 0.02, \alpha^+ = 12, \alpha^- =10, \beta^+=50, \beta^- = 5, \sigma=0.3, \theta=0.5, x=0.7$ and $T = 10$.}
\label{fig:swaptionsurface}
\end{figure}

\section*{Conclusion}
Classical interest rate market models are not capable of simultaneously allowing for semi-analytical pricing formulas for caplets and swaptions and guaranteeing nonnegative forward interest rates. One exception are the affine LIBOR models presented in \citet{KPT11}. This paper modifies their approach to also allow for driving processes which are not necessarily nonnegative. Caplet and swaption valuation is possible via one-dimensional numerical integration. This allows for a fast calculation of implied volatilities for these types of interest rate derivatives. With the additional flexibility of real-valued affine processes this type of model is capable of producing skewed implied volatility surfaces as well as implied volatility surfaces with pronounced smiles.

\begin{appendix}
\section{Proofs}
\begin{proof}[Proof of Lemma \ref{lem:monoton}]
For a function $f(x)$ denote its even and odd part by
$$f^e(x) = \frac{1}{2} (f(x) + f(-x)), \qquad f^o(x) = \frac{1}{2} (f(x) - f(-x)).$$
Note that if $f$ is monotonically increasing, the same is true for $f^o$.
Then \eqref{eq:coshmartingalefunction} can be written as
\begin{align*}
M_t^u(x) 
& = \frac{1}{2} \left(  \e^{\phi_{T-t}(u) + \psi_{T-t}(u) x}  +  \e^{\phi_{T-t}(-u) + \psi_{T-t}(-u) x} \right) \\
& =  \e^{\phi_{T-t}^e(u) + \psi_{T-t}^e(u) x} \cosh(\phi_{T-t}^o(u)+\psi_{T-t}^o(u) x)
\end{align*}
and
\begin{equation} \label{eq:Mufrac}
\frac{M_t^{u_{i}}(x)}{M_t^{u_{0}}(x)} =  \e^{(\phi_{T-t}^e(u_i) - \phi_{T-t}^e(u_{0}) )+ (\psi_{T-t}^e(u_i) -  \psi_{T-t}^e(u_{0})  ) x} \frac{\cosh(\phi_{T-t}^o(u_i)+\psi_{T-t}^o(u_i) x)}{\cosh(\phi_{T-t}^o(u_{0})+\psi_{T-t}^o(u_{0}) x)}.
\end{equation}
If $u_i = u_0$, then \eqref{eq:Mufrac} is constant and has no influence regarding monotonicity or maxima. Hence from now on assume $u_0 > u_i$ for all $i$. The function $g$ of equation \eqref{eq:monoton} can be written as
$$g(x) = \sum_{i=1}^n {c_i} \e^{A_i} \e^{a_i x} \frac{\cosh(B_i + b_i x)}{\cosh(B_0 + b_0 x)},$$
where for $i=0,\dots,n$
\begin{equation*}
\begin{aligned}
A_i & = (\phi_{T-t}^e(u_i) -  \phi_{T-t}^e(u_{0})), \qquad & B_i = \phi_{T-t}^o(u_i),  \\
a_i & = (\psi_{T-t}^e(u_i) -  \psi_{T-t}^e(u_{0})), & b_i = \psi_{T-t}^o(u_i). 
\end{aligned}
\end{equation*}

Since $\psi$ is monotonically increasing (see Lemma \ref{lem:psimontonicity}), also $\psi^o$ is monotonically increasing. With $\psi^o(0)=0$ it follows that $b_i \geq 0$ for all $i$. Furthermore note that $a_i < b_0 - b_i$ is equivalent to $\psi_{T-t}(u_i) < \psi_{T-t}(u_0)$ and $-a_i < b_0 - b_i$ is equivalent to $\psi_{T-t}(-u_0) < \psi_{T-t}(-u_i)$. Since $u_0 > u_i \geq 0$ the monotonicity of $\psi$ yields
\begin{equation} \vert a_i \vert < b_0 - b_i. \label{eq:aiest} \end{equation}

An elementary calculation gives
$$g^\prime(x) =  \frac{1}{ \cosh(B_0 + b_0 x)^2} \sum_{i=1}^n c_i \e^{A_i} \e^{a_i x} f_i(x), $$
where
\begin{align*} 
f_i(x) & = a_i \cosh(B_i+ b_i x)\cosh(B_0+b_0 x)+ b_i \sinh(B_i+b_i x) \cosh(B_0+b_0 x) \\
& - b_0 \cosh(B_i+b_i x) \sinh(B_0 +b_0 x).
\end{align*}
The derivative of $f_i$ is 
\begin{equation} \label{eq:gprime}
\begin{aligned} f_i^\prime(x)  =&   b_i \cosh(B_0 + b_0 x) \Big(a_i \sinh(B_i + b_i x) + (b_i - b_0) \cosh(B_i + b_i x)\Big) \\
 + &   b_0 \cosh(B_i + b_i x) \Big(a_i\sinh(B_0 + b_0 x) +  (b_i - b_0)  \cosh(B_0 + b_0 x)\Big) .
\end{aligned}
\end{equation}
Using \eqref{eq:aiest} the terms inside the brackets of each row in \eqref{eq:gprime} are strictly less than
$$ \vert a_i \vert (\sign{a_i} \sinh(B_j + b_j x) - \cosh(B_j + b_j x)) \leq 0, \qquad (j=i,0).$$
The last inequality is true since $\cosh(x) \pm \sinh(x) \geq 0$. The terms outside of the brackets in \eqref{eq:gprime} are all positive. Hence $f^\prime \geq 0$ and the $f_i$ are monotonically decreasing. Using \eqref{eq:aiest} a simple calculation shows that $ \lim_{x\rightarrow - \infty} f_i(x) =  \infty  \text{ and } \lim_{x\rightarrow  \infty} f_i(x) = -  \infty $. Since $c_i > 0$ for all $i$ the same is true for $\sum_{i=1}^n c_i \e^{a_i x} f_i(x)$. Hence $g$ has a single maximum and is decreasing to the left and right of it. Furthermore $g(x) \geq 0$ and again using \eqref{eq:aiest} $\lim_{x \rightarrow  \infty} g(x) = \lim_{x \rightarrow - \infty} g(x) = 0$.
\end{proof}

\begin{proof}[Proof of Lemma \ref{lem:fouriertransform}]
$f$ is continuous with compact support. Hence the extended Fourier transform $\hat{f}(z) = \int_{\R} f(x) \e^{- \i z x} \dd{x}$ exists for all $z \in \C$ and is analytic. For $z \neq 0, z \neq - \i v_k$, it is given by
\begin{align*}
\hat{f}(z)& = \int_{\kappa_1}^{\kappa_2} \e^{-\i z x} f(x) \dd{x} =  \frac{1}{\i z} \int_{\kappa_1}^{\kappa_2} \e^{-\i z x} f^\prime(x) \dd{x} \\
& = \frac{1}{\i z} \sum_k  {\frac{C_k v_k }{v_k - \i z}  { \left(\e^{(v_k - \i z) \kappa_2} - \e^{(v_k - \i z) \kappa_1 } \right) } }. 
\end{align*}
Since $\hat{f}(u - \i R) = O(u^{-2})$ for fixed $R$, it is absolutely integrable. By Fourier inversion 
\begin{equation*} \begin{split} f(x) & = \frac{1}{2 \pi} \int_{\mathrm{Im}(z) = - R} \e^{\i z x} \hat{f}(z) \dd{z} = \frac{1}{2 \pi} \int^\infty_0 \mathrm{Re}\left( \e^{(\i u + R) x} \hat{f}(u-\i R) \right) \dd{u},
\end{split} \end{equation*}
where the last equation follows from the fact that $f$ is real valued and the symmetry $\overline{\hat{f}(z)} = \hat{f}(-\overline{z})$. Since 
$\int \EV{\vert \e^{(\i z + R) X_t} \vert \vert \F_s}  \vert \hat{f}(z) \vert \dd{z} = \mathcal{M}_{X_t\vert X_s}(R) \int \vert \hat{f}(z) \vert \dd{z}$ is bounded if $R \in \mathcal{V} \cap \R$, conditional expectation and integration can be interchanged.
\end{proof}
\end{appendix}

\end{spacing}

\setlength{\bibsep}{0.0pt}

\end{document}